\documentclass[a4paper,11pt,oneside,english]{amsart}

\usepackage[latin9]{inputenc}
\usepackage{units}
\usepackage{amstext}
\usepackage{amsthm}
\usepackage{amssymb}
\usepackage{babel}
\usepackage{algorithm}
\usepackage{algpseudocode}
\usepackage{csquotes}

\usepackage[
  backend=biber,
  style = apa,
  maxcitenames=2,
  url=false,
  doi=true,
  dashed=false % don't use em-dash for repeated authors
]{biblatex}
\AtEveryBibitem{%
  \clearfield{note}%
  \clearfield{urldate}%
  \clearfield{accessed}%
  \clearfield{issn}
}
\DeclareFieldFormat{apacase}{#1}

\numberwithin{equation}{section}
\numberwithin{figure}{section}
\theoremstyle{definition}
\newtheorem{defn}{\protect\definitionname}
\theoremstyle{plain}
\newtheorem{prop}{\protect\propositionname}
\theoremstyle{plain}
\newtheorem{cor}{\protect\corollaryname}
\makeatother
\providecommand{\corollaryname}{Corollary}
\providecommand{\definitionname}{Definition}
\providecommand{\propositionname}{Proposition}

\newcommand\thankssymb[1]{\textsuperscript{\@fnsymbol{#1}}}

\begin{document}
\title{Observable consequences of mental accounting}

\author{Laura Blow \and Ian Crawford }

\thanks{Blow: School of Economics, University of Surrey, Guildford, Surrey, GU2 7XH, United Kingdom $\&$ IFS. l.blow@surrey.ac.uk. 
Crawford: Department of Economics, University of Oxford, Manor Road, Oxford, OX1 3UQ, United Kingdom; Nuffield College, New Road, Oxford, OX1 1NF, United Kingdom. ian.crawford@economics.ox.ac.uk.}

\begin{abstract}
We derive necessary and sufficient nonparametric conditions for several models of mental accounting. The paper characterises pure mental accounting, separable accounts, and labelled income, and compares these boundedly rational models with two rational multi-stage budgeting benchmarks. The resulting Afriat-style conditions make the observable implications of mental accounting explicit and refutable. In this sense, mental accounting is treated not as a loose description of behaviour, but as a formally refutable hypothesis about the organisation of consumption.
\end{abstract}

\maketitle
\section{Introduction}

Mental accounting is a broad notion which can encompass many kinds of mechanism used by individuals and households to organise, evaluate, and keep track of financial activities.  The existing literature encompasses both descriptive accounts of household budgeting and formal structural theories of particular mechanisms. This paper takes a different route: it studies a specific and economically tractable form of mental accounting: the assignment of consumption expenditures and income flows to latent accounts within a revealed-preference framework. We ask whether an observed sequence of prices and consumption choices can be rationalised as if the consumer were operating a system of mental accounts, and we derive necessary and sufficient nonparametric conditions for several variants of that idea, including pure mental accounting, separable accounts, labelled income, and rational two-stage budgeting. The result is a family of Afriat-style characterisations that make mental accounting empirically refutable, while also clarifying how standard utility maximisation and multi-stage budgeting appear as limiting cases.

The paper makes three contributions. First, it provides a revealed-preference characterisation of pure mental accounting with hidden partitions of goods into accounts, showing that the model admits a finite set of testable inequalities and nests GARP as a special case. Second, it extends the analysis to separable accounts and to labelled income, giving corresponding necessary and sufficient conditions for each structure. Third, it characterises weakly separable and additively separable two-stage budgeting as benchmark models, so that the observable implications of mental accounting can be compared directly with standard multi-stage consumer theory.

The paper sits at the intersection of three literatures. The first is the revealed-preference and nonparametric demand literature initiated by Afriat's Theorem (\cite{afriat1967construction} and developed further by \cite{Diewert1973}, and  \cite{varianNonparametricApproachDemand1982}), which shows how rationality can be characterised by a finite system of inequalities rather than by an assumed functional form. That approach is especially well suited to the present setting, where the object is not to estimate a structural utility function but to ask whether observed choices can be rationalised by a particular class of hidden budgeting rules. 

The second literature concerns multi-stage budgeting and separability, beginning with \cite{Strotz1955} and extended by \cite{Gorman_1959}. Strotz's discussion of budgeting across branches of expenditure and Gorman's characterisation of weak separability make clear that consumer choice may be organised through intermediate decision layers, and that such structure has testable implications for demand. The present paper builds on that tradition, but uses it to formulate a set of hidden-account models in which the relevant partitions are not observed and must therefore be inferred from the data. 

The third literature is based on Thaler's work on mental accounting (\cite{Thaler1999}). Thaler's account emphasises that households do not necessarily evaluate all expenditures within a single unified budget, but may instead organise spending into psychologically meaningful accounts. This paper takes that behavioural idea and places it inside the revealed-preference apparatus: rather than treating mental accounting as a loose descriptive label, it derives nonparametric necessary and sufficient conditions under which observed prices and quantities can be rationalised as if the consumer were operating with pure accounts, separable accounts, or labelled income. Recent work has provided formal mechanisms through which mental or personal budgets may arise from well-specified behavioural mechanisms. \cite{Koszegi_and_Matejka} derive category budgets from costly attention: when goods have a nested substitutability structure, consumers may attend to within-category relative shocks while holding total category spending fixed. \cite{Galperti_2019} instead gives a commitment-based theory in which present-biased consumers facing uncertainty about intratemporal trade-offs between current consumption goods may optimally impose good-specific spending caps, sometimes alongside a minimum-savings rule. These papers provide structural explanations for why consumers may use budgets---and in doing so necessarily impose structure on preferences, information, or commitment possibilities. Our approach is different but complementary: we do not take a stand on the mechanism by which accounts are formed, but instead derive nonparametric revealed-preference conditions under which observed choices are consistent with latent account-specific constraints.

In what follows, Proposition 1 characterises pure mental accounting and shows that hidden budget accounts relax global rationality while remaining testable for fixed intermediate numbers of groups. Proposition 2 considers separable accounts, where the account partition also defines independent preference domains and hence requires GARP within each account. Proposition 3 characterises labelled income, where labels act as soft earmarks and a multiplier captures whether the label is behaviourally binding. Propositions 4 and 5 then describe rational two-stage budgeting benchmarks, in which the allocation of expenditure across groups must itself satisfy upper-stage optimality restrictions.

\section{Boundedly Rational Mental Accounting}

In this section we characterise some boundedly rational mental accounting models, including pure accounts, separable accounts, and labelling. We suppose that we observe, for an individual agent, a sequence of consumption decisions and the prices at which she transacted $\left\{( \mathbf{p}_{t},\mathbf{q}_{t})\right\}
_{t=1,...,T}$.  The total budget is denoted by $x_t=\mathbf{p}_{t}\cdot \mathbf{q}_{t}$. We assume throughout that the unit of time over which the accounts are operated correspond to the definition of a period/observation in our data. This might not be the case if food budgets (for instance) are balanced weekly, but clothing budgets are balance monthly. 

\subsection{Pure mental accounting}

The individual has standard preferences over a $K-$vector of goods described by $u(\mathbf{q})$. These preferences may be separable but there is no presumption that this is the case.  Her budget is divided into several $(M)$ accounts $\left\{ x_{t}^{m}\right\} _{t=1,...,T}^{m=1,...,M}$ which must sum to the total available budget. We take it that the ``last'' budget is a residual (petty cash) account: $x_{t}^{M}=x_{t}-\sum_{m=1}^{M-1}x_{t}^{m}$. The balance of each account may vary both between accounts and over time. Each item in the commodity vector is allocated to a specific account; thus we denote the vector of items in the $m$'th account by $\mathbf{q}^{m}$. The number of items in each account may differ across accounts. These groups of items form a \textit{partition} of the commodity vector: $\mathbf{q}=[%
\mathbf{q}^{1},...,\mathbf{q}^{M}]$. The allocation of items to accounts does not vary over time (e.g. food items are always assigned to the food account).

The optimising model is
\[
\max_{\mathbf{q}}u(\mathbf{q}) \text{ subject to } \mathbf{p}_{t}^{1}\cdot\mathbf{q}^{1}\leq x_{t}^{1},...,\mathbf{p}_{t}^{M}\cdot\mathbf{q}^{M}\leq x_{t}-\sum_{m=1}^{m=M-1}x_{t}^{m}
\]
We suppose that only the prices and the commodity vectors are observed. In particular, neither the allocation of items to accounts nor the account balances themselves ($\left\{ x_{t}^{m}\right\} _{t=1,...,T}$ ) are observed.

\begin{defn}
A mental accounting model with $M$ accounts rationalises the data $\left\{( 
\mathbf{p}_{t},\mathbf{q}_{t})\right\} _{t=1,...,T}$ if there exists a concave, continuous, monotonic, function $u$ and a set of accounts $%
\{x_{t}^{1},...,x_{t}^{M}\}_{t=1,...,T}$ such that $u(\mathbf{q}_{t})\geq u(%
\mathbf{q})$ for all $\mathbf{q}$ such that $\mathbf{p}_{t}^{m}\cdot \mathbf{q}^{m}\leq x_{t}^{m}$ and $\sum_{m=1}^{m=M}x_{t}^{m}=x_{t}$ for all $t$  and $m$
\end{defn}

This is a statement of the principle of revealed preference: that the utility function must assign a higher value to the selected bundle than any other which satisfies the accounts and is affordable. The conditions for this model are given in the following Proposition.

\begin{prop}
(Pure mental accounting). The following statements are equivalent.\newline
1. There exists a concave, continuous, monotonic utility function and a set
of $M$ accounts which rationalise the data.\newline
2. There exists a partition of items into $M$ groups such that there exists
a set of real numbers $\{U_{t},\lambda_{t}^{1},...,\lambda_{t}^{M}%
\}_{t=1,...,T}$ such that 
\begin{eqnarray*}
U_{s} & \leq & U_{t}+\sum_{m=1}^{m=M}\lambda_{t}^{m}\mathbf{p}_{t}^{m}\cdot(\mathbf{q}_{s}^{m}-\mathbf{q}_{t}^{m})\,\,\,\,\forall s,t \\
\lambda_{t}^{m} & > & 0\,\,\,\,\,\forall m,t
\end{eqnarray*}
\end{prop}
\medskip
\begin{proof}
The proof of this and subsequent results are in the Appendix.
\end{proof}
\medskip
Economically, this says that the consumer behaves as if goods are split into hidden accounts, each with its own shadow budget, and the observed choice is jointly optimal only after that latent partition is taken into account.

This result provides necessary and sufficient conditions for the model: if a suitable partition exists then the model can provide a rationalisation for the data and is a potential data-generating process; if not then the data are incompatible with mental accounting as described by the optimising model above.

We have the following corollaries which relate to the number of accounts.

\begin{cor}
Proposition 1 with $M=1$ is equivalent to GARP.
\end{cor}

\begin{cor}
If it is the case that $M=K$ then the data must satisfy the conditions in Proposition 1.
\end{cor}

If all of the data are in a single account then the model is the standard utility maximisation model and the conditions in Proposition 1 reduce to GARP. Thus the standard rational choice model is a special case of the mental accounting model. At the other end of the spectrum if there are as many accounts as goods then the model is unfalsifiable.

The parameter $M$ therefore governs the empirical content of the model. At one extreme, $(M=1)$ imposes the same discipline as the standard unitary model. At the other extreme, $(M=K)$ leaves enough freedom in the account-specific shadow prices to rationalise any finite data set. Intermediate values of $M$ are the economically interesting case: they allow violations of global GARP to be explained by account-specific budgeting, but only when those violations can be organised around a stable latent partition of goods.

\subsection{Separable mental accounting}

An interesting variant of the mental accounting model assumes that the accounts are defined over groups of goods which are separable and hence that the marginal rates of substitution between items within an account are independent of the consumption of items in other accounts. Then the consumer solves the following problem. 
\begin{equation*}
\max_{\mathbf{q}}u(v^{1}(\mathbf{q}^{1}),...,v^{M}(\mathbf{q}^{M}))
\text{ subject to } \mathbf{p}_{t}^{1}\cdot\mathbf{q}^{1}\leq x_{t}^{1},...,\mathbf{p}_{t}^{m}\cdot\mathbf{q}^{M}\leq x_{t}-\sum_{m=1}^{m=M-1}x_{t}^{m}
\end{equation*}
A collective household model in which each member was individually responsible for the purchase of certain classes of items would exactly reflect this structure (this is akin to \cite{Belgians2007} \textquotedblleft situation-dependent dictatorship\textquotedblright\ but with the partition of responsibility being between goods rather than periods). This is also the case considered by \cite{Strotz1955}. In this case the only restriction imposed by maximising behaviour is that the data within each account satisfies GARP.

\begin{prop}
(Separable accounting). The following statements are equivalent.\\
1. There exists a concave, continuous, monotonic utility function and a set of $M$ separable accounts which rationalise the data.\\
2. There exist a partition of items into $M$ groups such that there exists a set of real numbers $\{U_{t}^{1}...,U_{t}^{M},\lambda _{t}^{1},...,\lambda
_{t}^{M}\}_{t=1,...,T}$ such that 
\begin{eqnarray*}
U_{s}^{m} &\leq &U_{t}^{m}+\lambda _{t}^{m}\mathbf{p}_{t}^{m}\cdot(\mathbf{q}_{s}^{m}-\mathbf{q}_{t}^{m})\,\,\,\,\forall m,s,t \\
\lambda _{t}^{m} &>&0\,\,\,\,\,\forall m,t
\end{eqnarray*}
3. There exist a partition of items into $M$ groups such that the data \\$%
\left\{ (\mathbf{p}_{t}^{m},\mathbf{q}_{t}^{m})\right\} _{t=1,...,T}$ satisfies GARP for each partition $m=1,...,M$.
\end{prop}

Here the accounts behave like independent sub-problems, so empirical rationalisability requires only that each account's data satisfy GARP on its own.

The separable version changes the interpretation of accounts. In Proposition 1, accounts restrict feasible reallocations, but preferences may still link goods across accounts. In Proposition 2, by contrast, the accounts are also preference domains: the marginal trade-offs among goods within one account are independent of consumption in other accounts. The testable implication is correspondingly local. Once a candidate partition has been chosen, the data for each account must be rationalisable as a separate consumer problem. This makes separable accounting more transparent empirically, but also gives the account partition a stronger economic interpretation.

Proposition 2 implies Proposition 1. Thus we have

\begin{cor}
If the data satisfy the condition for separable accounting they necessarily satisfy the conditions for pure mental accounts.
\end{cor}

\subsection{Labelling}

Suppose now that one part of the consumer\textquoteright s budget is exogenously labelled, but not mandated, as being intended for spending on a certain category of goods. Real world examples from the UK include Child Benefit\footnote{Child Benefit is a government-funded payment in the UK designed to
help parents and guardians with the costs of raising a child.} and the Winter Fuel Payment.\footnote{The Winter Fuel Payment, is an annual, tax-free lump sum paid by the
UK government which designed to help older people with heating costs
during the colder months.} Unlike the Supplemental Nutrition Assistance Program in the US (food stamps) which provides targeted food-purchasing assistance for low-income households, recipients of Child Benefit and the Winter Fuel Payment receive the value in cash and have the autonomy to spend it on whatever they determine. Call this labelled part of the budget $x^{A}$. A rational consumer would ignore the labelling but suppose the consumer operating a mental account feels it necessary to spend the labelled portion only on goods in group $A$. The rest of the goods we will denote by $B$. If a consumer operates according to such a model then we have the following result.\bigskip{}

\begin{prop}
(Labelling). The following statements are equivalent.\\
1. There exists a concave, continuous, monotonic utility function which rationalises the data $\left\{ (\mathbf{p}_{t},\mathbf{q}_{t},x_{t}^{A})\right\} _{t=1,...,T}$.\\
2. There exists a set of real numbers $\{U_{t},\lambda_{t},\mu_{t}\}_{t=1,...,T}$ such that 
\begin{align*}
U_{s} & \leq U_{t}+\lambda_{t}\mathbf{p}_{t}\cdot(\mathbf{q}_{s}-\mathbf{q}_{t})+\mu_{t}\mathbf{p}_{t}^{A}\cdot(\mathbf{q}_{s}^{A}-\mathbf{q}_{t}^{A})\\
\lambda_{t} & >0\\
\mu_{t} & \ge0\text{ (with equality when $\mathbf{p}_{t}^{A}\cdot\mathbf{q}_{t}^{A}>x_{t}^{A}$}).
\end{align*}
\end{prop} 
The multiplier $(\mu_t)$ measures the behavioural force of the label. When observed spending on the labelled category exceeds the labelled amount, the labelled income is infra-marginal: the consumer would have spent at least that much on the category even without the label, so the labelling restriction has no additional shadow value. In that case $\mu_t=0$, and the revealed-preference inequalities reduce locally to the standard Afriat inequalities. When spending on the labelled category is exactly equal to the labelled amount, the label may bind. The term involving $\mu_t$ then allows the labelled category to carry an additional shadow value, capturing the idea that the consumer treats labelled income as earmarked even though it is formally fungible.

\section{Rational Mental Accounting}

In this section we turn to multi-stage budgeting as the rational benchmark. Two-stage, or multi-stage budgeting is very close in spirit to some boundedly rational mental accounting but multi-stage budgeting requires that the overall allocation of expenditures is consistent with maximisation - i.e. not just that it is optimal, conditional on the initial allocation of income to accounts. \cite{Strotz1955} argued that a sufficient condition for two-stage budgeting is that the household's utility function be separable. \cite{Gorman_1959} showed that, while necessary, separability is not sufficient. In addition, it is required that the sub-utility functions enter utility either additively or through an intermediate function which is homogeneous of degree one. These restrictions on the consumer's preferences imply empirically refutable restrictions on the system of demand functions.

\begin{prop}
(Weakly separable two-stage budgeting). The following statements are equivalent:\\
1. There exists a weakly separable utility function with homothetic sub-utilities which provide a two-stage budgeting rationalisation for the data.\\
2. There exist real numbers $\{U_{t}^{m},\lambda_{t}^{m}\}_{t=,...,T}^{m=1,...,M}$ such that
\begin{eqnarray*}
U_{t}^{m} &\leq& U_{s}^{m}+\lambda_{t}^{m}\mathbf{p}_{t}^{m}\cdot\left(\mathbf{q}_{s}^{m}-\mathbf{q}_{t}^{m}\right)\;\forall\;s,t \\
U_{t}^{m} &=&\lambda_{t}^{m}\mathbf{p}_{t}^{m}\cdot\mathbf{q}_{t}^{m}  \; \forall\;t\\
\lambda_{t}^{m} &>& 0\;\;\;\;\forall\;t\\
\end{eqnarray*}
and $\{U_{t}^{m},\tfrac{1}{\lambda_{t}^{m}}\}_{t=1,...,T}^{m=1,...,M}$ satisfies GARP.
\end{prop}

Here both stages of choice are rational: the within-account allocation satisfies an Afriat-style condition, and the upper-level budgeting problem is itself consistent with rational choice. The following result covers the additive case.

\begin{prop}
(Additive two-stage budgeting). The following statements are equivalent:\\
1. There exists an additively separable utility function with sub-utilities of the Gorman Polar Form which provide a two-stage budgeting rationalisation for the data.\\
2. There exist real numbers $\{U_{t}^{m},\lambda_{t}\}_{t=,...,T}^{m=1,...,M}$ such that
\begin{eqnarray*}
U_{t}^{m}& \leq& U_{s}^{m}+\lambda _{t}\mathbf{p}_{t}^{m}\cdot\left( 
\mathbf{q}_{s}^{m}-\mathbf{q}_{t}^{m}\right) \;\;\;\;\forall \;s,t \\
\lambda _{t}& >&0\;\;\;\;\forall \;t.
\end{eqnarray*}
\end{prop}

This is the additive benchmark, in which a common outer shadow price governs the stage-level allocations, yielding a simpler but still fully rational multi-stage budgeting structure. Both of these results link directly to the previous proposition regarding boundedly rational separable mental accounting.

\begin{cor}
If the data satisfy the condition for two-stage budgeting they necessarily satisfy the conditions for separable accounts.
\end{cor}

The comparison with Propositions 1 and 2 separates two reasons why account-level behaviour may look rational. Under boundedly rational mental accounting, the allocation of expenditure to accounts is taken as given from the consumer's point of view. Under rational two-stage budgeting, that allocation must itself be rationalisable as part of the consumer's overall optimisation problem. The additional upper-stage restrictions in Propositions 4 and 5 are therefore what distinguish genuine two-stage rationality from merely account-by-account consistency.

\section{Implementation}

Whilst the empirical implementation of the labelling model is straightforward given observability of the labelled budget, the Pure Mental Accounting, Separable Accounts and the two versions of Rational Accounting share the fact that the grouping of goods is latent.  The structure of the empirical problem in all of these cases has two parts: checking model-specific feasibility for a given partition, and searching over partitions. The latter is finite but combinatorially large, since the number of partitions grows very fast with the number of goods.

For relatively small problems exhaustive enumeration of the partitions is computationally feasible.\footnote{Exactly what constitutes a `small problem' will depend on the computing power at hand, but to illustrate it may be useful to note that for 10, 11 and 12 goods respectively, the number of partitions is given by the 10th, 11th and 12th Bell Numbers which are 115,975, 678,570 and 4,213,597 respectively.}  In such cases, each candidate partition may be checked using the relevant model-specific feasibility test, and among those that satisfy the feasibility, one may select the partition with the smallest Selten Area, a criterion which measures the restrictiveness of the feasibility condition; see \cite{selten1991properties} and \cite{beattyHowDemandingRevealed2011}. 

\begin{algorithm}[H]
\caption{Greedy local search over latent partitions}
\label{alg:greedy_partition_search}
\begin{algorithmic}[1]
\Require Data $\{(\mathbf{p}_t,\mathbf{q}_t)\}_{t=1,...,T}$, initial partition $P_0$, model index $j \in \{1,2,4,5\}$, feasibility test $F_j(\cdot)$, selection criterion $SA(\cdot)$
\State Set $P \gets P_0$
\State Compute $SA(P)$ if $F_j(P)=1$; otherwise set $SA(P)=+\infty$
\Repeat
    \State Generate the neighbourhood $\mathcal{N}(P)$ of $P$ consisting of all partitions obtained by a single-good reassignment move, followed by relabelling
    \State Test each $P' \in \mathcal{N}(P)$ using the model-specific feasibility test $F_j(P')$
    \State Let   $P^\ast \in \arg\min \{ SA(P') : P' \in \mathcal{N}(P),\ F_j(P')=1 \}$
    \If{there exists a feasible $P^\ast$ with $SA(P^\ast) < SA(P)$}
        \State Set $P \gets P^\ast$
        \State Update $SA(P)$
    \Else
        \State Stop
    \EndIf
\Until{no feasible neighbouring partition delivers a strict improvement in $SA(\cdot)$}
\State \Return $P$, together with the associated feasibility certificate and $SA(P)$
\end{algorithmic}
\end{algorithm}

For larger problems brute force is not feasible. A practical heuristic is to implement a greedy local search over partitions. Beginning from an initial partition, generate the neighbourhood defined by single-good reassignment moves and relabelling,\footnote{By \textit{relabelling} we mean that account labels are purely notational and are replaced by a unique standard numbering, so that the same partition is not counted multiple times under different labels. For example, the membership vectors (2, 2, 1, 1, 3) and (1, 1, 2, 2, 3) represent the same partition of five goods into three accounts; this is \{{1, 2\}, \{3, 4\}, \{5\}}, or equivalently \{{3, 4\}, \{1, 2\}, \{5\}}, since block order is immaterial. After relabelling, both are written in the same form (1, 1, 2, 2, 3). } test each neighbour using the relevant model-specific feasibility condition, and move to the feasible neighbour with the lowest Selten Area (denoted $SA(\cdot)$ in the algorithm) whenever this improves on the current partition. The process is iterated until no feasible neighbouring partition delivers a strict improvement. Algorithm 1 provides suitable pseudo-code. This search template applies to Propositions 1, 2, 4 and 5; only the feasibility test $F_j(\cdot)$ differs across models. 

As well as being a practical option when there are many goods, the benefit of this procedure is that by seeding the search with an economically sensible starting partition it will help avoid the problem of searching over nonsense partitions which may satisfy the conditions but be economically implausible (for example, partitions which place chalk and cheese in the same group). We emphasise that this procedure is intended purely as a heuristic computational aid; the theoretical characterisations of each accounting model in the paper do not depend on it.

\section{Conclusions}

We have characterised three boundedly rational models of mental accounting ---pure accounts, separable accounts, and labelled income---together with two rational benchmarking models of multi-stage budgeting. The common feature of these characterisations is that each model has a finite set of observable revealed-preference implications, so that mental accounting can be treated as a refutable economic hypothesis rather than a purely descriptive idea. The paper therefore connects the behavioural literature on mental accounting with the nonparametric demand literature in a way that is deliberately narrow, but formally precise. In that sense, the contribution is not a general theory of mental accounting (which is a very broad church), but a set of testable characterisations for some specific economically useful versions of it.

\printbibliography

\pagebreak

\section*{Appendix - Proofs}

\section*{Proof of Proposition 1.}

Necessity. Assuming differentiability (alternatively the following arguments can be established using super-gradients) the first order conditions are 
\[
\nabla u(\mathbf{q}_{t}^{m})\leq \lambda_{t}^{m}\mathbf{p}_{t}^{m}\,\,\,\,\,\forall m,t
\]
Concavity implies 
\[
u(\mathbf{q}_{s}^{1},...,\mathbf{q}_{s}^{m})\leq u(\mathbf{q}_{t}^{1},...,\mathbf{q}_{t}^{m})+\sum_{m=1}^{m=M}\nabla u(
\mathbf{q}_{t}^{m})\cdot(
\mathbf{q}_{s}^{m}-\mathbf{q}_{t}^{m})
\]
and substituting in the first order conditions gives

\[
u(\mathbf{q}_{s}^{1},...,\mathbf{q}_{s}^{m})\leq u(\mathbf{q}_{t}^{1},...,%
\mathbf{q}_{t}^{m})+\sum_{m=1}^{m=M}\lambda _{t}^{m}\mathbf{p}_{t}^{m}\cdot (\mathbf{q}_{s}^{m}-
\mathbf{q}_{t}^{m})
\]
Since $u$ is real-valued there exist real numbers $\{U_{t},\lambda_{t}^{1},...,\lambda_{t}^{M}\}_{t=1,...,T}$ and a partition $\mathbf{q}_{t}=[\mathbf{q}_{t}^{1},...,\mathbf{q}_{t}^{M}]$ such that 

\begin{eqnarray*}
U_{s} &\leq &U_{t}+\sum_{m=1}^{m=M}\lambda _{t}^{m}\mathbf{p}_{t}^{m}\cdot(\mathbf{q}_{s}^{m}-\mathbf{q}_{t}^{m})\,\,\,\,\forall s,t \\
\lambda _{t}^{m} &>&0\,\,\,\,\,\forall m,t
\end{eqnarray*}
Sufficiency. Now suppose that there exist numbers and a partition which satisfy these inequalities. Define a utility function $U(\mathbf{q})$ by 
\begin{equation*}
U(\mathbf{q})=\min_{i}\left\{U_{i}+\sum_{m=1}^{m=M}\lambda _{i}^{m}\mathbf{p}%
_{i}^{m}\cdot(\mathbf{q}^{m}-\mathbf{q}_{i}^{m})\right\}
\end{equation*}

Note that this is piece-wise linear, monotonic, continuous and concave.
Consider $\mathbf{q}_{t}$ and the partition $\mathbf{q}_{t}=[\mathbf{q}%
_{t}^{1},...,\mathbf{q}_{t}^{M}]$ 
\begin{equation*}
U(\mathbf{q}_{t})=\min_{i}\left\{U_{i}+\sum_{m=1}^{m=M}\lambda _{i}^{m}\mathbf{p}%
_{i}^{m}\cdot(\mathbf{q}_{t}^{m}-\mathbf{q}_{i}^{m})\right\}
\end{equation*}
When $i=t$ then $U_{i}+\sum_{m=1}^{m=M}\lambda _{i}^{m}\mathbf{p}%
_{i}^{m}\cdot(\mathbf{q}_{t}^{m}-\mathbf{q}_{i}^{m})=U_{t}$ and we have (by construction) 

\begin{equation*}
U_{t}\leq U_{i}+\sum_{m=1}^{m=M}\lambda _{i}^{m}\mathbf{p}_{i}^{m}\cdot(\mathbf{q}_{t}^{m}-\mathbf{q}_{i}^{m})\,\,\,\,\forall i
\end{equation*}
Hence $U(\mathbf{q}_{t})=U_{t}$. Now consider some arbitrary $\mathbf{q}$ and the partition $\mathbf{q}=[\mathbf{q}^{1},...,\mathbf{q}^{M}]$ such that $\mathbf{p}_{t}^{m}\cdot\mathbf{q}^{m}\leq x_{t}^{m}$ for all $m.$ We have 
\begin{equation*}
U(\mathbf{q})=\min_{i}\left\{U_{i}+\sum_{m=1}^{m=M}\lambda _{i}^{m}\mathbf{p}_{i}^{m}\cdot(\mathbf{q}^{m}-\mathbf{q}_{i}^{m})\right\}
\end{equation*}
We have 
\begin{equation*}
\mathbf{p}_{t}^{m}\cdot(\mathbf{q}^{m}-\mathbf{q}_{t}^{m})\leq 0\,\,\,\forall m
\end{equation*}
and hence since $\lambda _{t}^{m}>0$ we have 
\begin{equation*}
\lambda _{t}^{m}\mathbf{p}_{t}^{m}\cdot(\mathbf{q}^{m}-\mathbf{q}_{t}^{m})\leq
0\,\,\,\forall m
\end{equation*}
and 
\begin{equation*}
\sum_{m=1}^{m=M}\lambda _{t}^{m}\mathbf{p}_{t}^{m}\cdot(\mathbf{q}^{m}-\mathbf{q}_{t}^{m})\leq 0
\end{equation*}
Thus 
\begin{equation*}
U_{t}\geq U_{t}+\sum_{m=1}^{m=M}\lambda _{t}^{m}\mathbf{p}_{t}^{m}\cdot(
\mathbf{q}^{m}-\mathbf{q}_{t}^{m})
\end{equation*}
and hence 
\begin{equation*}
U_{t}\geq U_{t}+\sum_{m=1}^{m=M}\lambda _{t}^{m}\mathbf{p}_{t}^{m}\cdot(
\mathbf{q}^{m}-\mathbf{q}_{t}^{m})\geq
\min_{i}\left\{U_{i}+\sum_{m=1}^{m=M}\lambda _{i}^{m}\mathbf{p}_{i}^{m}\cdot(\mathbf{q}^{m}-\mathbf{q}_{i}^{m})\right\}=U(\mathbf{q})
\end{equation*}
Thus $U(\mathbf{q}_{t})=U_{t}\geq U(\mathbf{q}).$$\blacksquare $

\section*{Proof of Corollary 1}

The condition in Proposition 1 reduces to 
\begin{equation*}
U_{t}\geq U_{t}+\lambda_{t}^{1}\mathbf{p}_{t}^{1}\cdot(
\mathbf{q}_{s}^{1}-\mathbf{q}_{t}^{1})
\end{equation*}
with $\mathbf{p}_{t}^{1}=\mathbf{p}_{t}$and $\mathbf{q}_{t}^{1}=\mathbf{q}%
_{t}$ for all $t$. By Afriat's Theorem this is equivalent to GARP. $%
\blacksquare$

\section*{Proof of Corollary 2}

With a single good, each in its own account, the condition in Proposition 1 
\begin{equation*}
U_{s}\leq U_{t}+\sum_{m=1}^{m=M}\lambda_{t}^{m}p_{t}^{m}\cdot(\mathbf{{q}}_{s}^{m}-\mathbf{q}_{t}^{m})
\end{equation*}
can be written as 
\begin{equation*}
U_{s}\leq U_{t}+\mu_{t}\boldsymbol{\pi}_{t}\cdot(\mathbf{q}_{s}-\mathbf{q}_{t})
\end{equation*}
where 
\begin{equation*}
\lambda_{t}^{m}p_{t}^{m}=\pi_{t}^{m}
\end{equation*}
This is is equivalent to the requirement that there exist some virtual prices $\{\boldsymbol{\pi}_{t}\}_{t=1,...,T}$ such that the data $\{(\boldsymbol{\pi}_{t},\mathbf{q}_{t})\}_{t=1,...,T}$ satisfies GARP. Setting all virtual prices to one suffices to show that such prices always exist.$\blacksquare$

\section*{Proof of Proposition 2}

Weak separability is necessary and sufficient for the second (lower) stage of two-stage budgeting. This implies GARP within each separable group (\cite{varian1983non}, Theorem 3) and hence the existence of sub-utility functions
which rationalise the model 
\begin{equation*}
max_{\mathbf{q}}v^{m}(\mathbf{q}^{m})\;\mbox{subject to}.\;\mathbf{p}%
_{t}^{m}\cdot\mathbf{q}^{m}\leq x_{t}^{m}
\end{equation*}
Summing the corresponding Afriat Inequalities 
\begin{equation*}
U_{s}^{m}\leq U_{t}^{m}+\lambda_{t}^{m}\mathbf{p}_{t}^{m}\cdot(=
\mathbf{q}_{s}^{m}-\mathbf{q}_{t}^{m})
\end{equation*}
over the groups gives 
\begin{equation*}
\sum_{m=1}^{M}U_{s}^{m}\leq\sum_{m=1}^{M}U_{t}^{m}+\sum_{m=1}^{M}\lambda_{t}^{m}\mathbf{p}_{t}^{m}\cdot(\mathbf{q}_{s}^{m}-\mathbf{q}_{t}^{m})
\end{equation*}
or, setting $U_{t}=\sum_{m=1}^{M}U_{t}^{m}$, we can write this as 
\begin{equation*}
U_{s}\leq U_{t}+\sum_{m=1}^{M}\lambda_{t}^{m}\mathbf{p}_{t}^{m}\cdot(%
\mathbf{q}_{s}^{m}-\mathbf{q}_{t}^{m})
\end{equation*}
The rest of the proof parallels Proposition 1.$\blacksquare$

\section*{Proof of Corollary 3}

Immediate from inspection of the conditions in Propositions 1 and 2.$%
\blacksquare$

\section*{Proof of Proposition 3.}

Necessity. Arrange the commodity vector so that
\[
\mathbf{q}=\left[\begin{array}{c}
\mathbf{q}^{A}\\
\mathbf{q}^{B}
\end{array}\right]
\]
The consumer's optimisation problem is
\[
\max_{\mathbf{q}}U(\mathbf{q})\text{ subject to }\mathbf{p}_{t}\cdot\mathbf{q}\le x_{t};\mathbf{p}_{t}^{A}\cdot\mathbf{q}^{A}\ge x_{t}^{A}
\]
Assuming differentiability (alternatively the following arguments
can be established using super-gradients) the first order conditions
are 
\[
\left[\begin{array}{c}
\nabla_{A}U(\mathbf{q}_{t})\\
\nabla_{B}U(\mathbf{q}_{t})
\end{array}\right]\le\left[\begin{array}{c}
\lambda_{t}\mathbf{p}_{t}^{A}+\mu_{t}\mathbf{p}_{t}^{A}\\
\lambda_{t}\mathbf{p}_{t}^{B}
\end{array}\right]
\]
where $\lambda_{t}>0$ and $\mu_{t}\ge0$ with equality when the spending
constraint on the labelled goods does not bind ($\mathbf{p}_{t}^{A}\cdot\mathbf{q}^{A}>x_{t}^{A})$. 

\noindent Concavity of $u$ implies
\[
u(\mathbf{q}_{s})\le u(\mathbf{q}_{t})+\nabla u(\mathbf{q}_{t})\cdot(\mathbf{q}_{s}-\mathbf{q}_{t})
\]
 Substituting in the first order conditions preserves the inequality
giving
\[
u(\mathbf{q}_{s})\le u(\mathbf{q}_{t})+\left(\lambda_{t}\mathbf{p}_{t}^{A}+\mu_{t}\mathbf{p}_{t}^{A}\right)\cdot(\mathbf{q}_{s}^{A}-\mathbf{q}_{t}^{A})+\left(\lambda_{t}\mathbf{p}_{t}^{B}\right)\cdot(\mathbf{q}_{s}^{B}-\mathbf{q}_{t}^{B})
\]
or
\begin{align*}
u(\mathbf{q}_{s}) & \le u(\mathbf{q}_{t})+\lambda_{t}\mathbf{p}_{t}^{A}\cdot(\mathbf{q}_{s}^{A}-\mathbf{q}_{t}^{A})+\mu_{t}\mathbf{p}_{t}^{A}\cdot(\mathbf{q}_{s}^{A}-\mathbf{q}_{t}^{A})+\lambda_{t}\mathbf{p}_{t}^{B}\cdot(\mathbf{q}_{s}^{B}-\mathbf{q}_{t}^{B})\\
u(\mathbf{q}_{s}) & \le u(\mathbf{q}_{t})+\lambda_{t}\mathbf{p}_{t}^{A}\cdot(\mathbf{q}_{s}^{A}-\mathbf{q}_{t}^{A})+\lambda_{t}\mathbf{p}_{t}^{B}\cdot(\mathbf{q}_{s}^{B}-\mathbf{q}_{t}^{B})+\mu_{t}\mathbf{p}_{t}^{A}\cdot(\mathbf{q}_{s}^{A}-\mathbf{q}_{t}^{A})\\
u(\mathbf{q}_{s}) & \le u(\mathbf{q}_{t})+\lambda_{t}\mathbf{p}_{t}\cdot(\mathbf{q}_{s}-\mathbf{q}_{t})+\mu_{t}\mathbf{p}_{t}^{A}\cdot(\mathbf{q}_{s}^{A}-\mathbf{q}_{t}^{A})
\end{align*}
Finally since $u$ is real valued there must exists real numbers $\left\{ U_{t},\lambda_{t},\mu_{t}\right\} _{t=1,...,T}$such
that
\begin{align*}
U_{s} & \leq U_{t}+\lambda_{t}\mathbf{p}_{t}\cdot(\mathbf{q}_{s}-\mathbf{q}_{t})+\mu_{t}\mathbf{p}_{t}^{A}\cdot(\mathbf{q}_{s}^{A}-\mathbf{q}_{t}^{A})\\
\lambda_{t} & >0\\
\mu_{t} & \ge0\text{ with equality when \ensuremath{\mathbf{p}_{t}^{A}\cdot\mathbf{q}_{t}^{A}>x_{t}^{A}}}
\end{align*}
Sufficiency. Now suppose that there exist numbers which satisfy these
inequalities. Define a utility function $U(\mathbf{q})$ by 
\[
U(\mathbf{q})=\min_{i\in\{1,...,T\}}\left\{ U_{i}+\left(\lambda_{i}\mathbf{p}_{i}^{A}+\mu_{i}\mathbf{p}_{i}^{A}\right)\cdot(\mathbf{q}^{A}-\mathbf{q}_{i}^{A})+\lambda_{i}\mathbf{p}_{i}^{B}\cdot(\mathbf{q}^{B}-\mathbf{q}_{i}^{B})\right\} 
\]
Note that this is piece-wise linear, monotonic, continuous and concave.
Consider $\mathbf{q}_{t}.$ Using the utility function above
\[
U(\mathbf{q}_{t})=\min_{i\in\{1,...,T\}}\left\{ U_{i}+\left(\lambda_{i}\mathbf{p}_{i}^{A}+\mu_{i}\mathbf{p}_{i}^{A}\right)\cdot(\mathbf{q}_{t}^{A}-\mathbf{q}_{i}^{A})+\lambda_{i}\mathbf{p}_{i}^{B}\cdot(\mathbf{q}_{t}^{B}-\mathbf{q}_{i}^{B})\right\} 
\]
By construction
\[
U_{t}\le U_{i}+\left(\lambda_{i}\mathbf{p}_{i}^{A}+\mu_{i}\mathbf{p}_{i}^{A}\right)\cdot(\mathbf{q}_{t}^{A}-\mathbf{q}_{i}^{A})+\lambda_{i}\mathbf{p}_{i}^{B}\cdot(\mathbf{q}_{t}^{B}-\mathbf{q}_{i}^{B})\;\forall i\in\{1,...,T\}
\]
and hence the bundle $\mathbf{q}_{t}$ is assigned the value
\[
U(\mathbf{q}_{t})=U_{t}
\]
Now consider an arbitrary bundle $\mathbf{q}$such that $\mathbf{p}_{t}\cdot\mathbf{q}_{t}\ge\mathbf{p}_{t}\cdot\mathbf{q}$.
We need to show that $U(\mathbf{q}_{t})\ge U(\mathbf{q})$.

\noindent The utility function assigns
\[
U(\mathbf{q})=\min_{i\in\{1,...,T\}}\left\{ U_{i}+\left(\lambda_{i}\mathbf{p}_{i}^{A}+\mu_{i}\mathbf{p}_{i}^{A}\right)\cdot(\mathbf{q}^{A}-\mathbf{q}_{i}^{A})+\lambda_{i}\mathbf{p}_{i}^{B}\cdot(\mathbf{q}^{B}-\mathbf{q}_{i}^{B})\right\} 
\]
or
\[
U(\mathbf{q})=\min_{i\in\{1,...,T\}}\left\{ U_{i}+\lambda_{i}\mathbf{p}_{i}\cdot(\mathbf{q}-\mathbf{q}_{i})+\mu_{i}\mathbf{p}_{i}^{A}\cdot(\mathbf{q}^{A}-\mathbf{q}_{i}^{A})\right\} 
\]
Consider the element where $i=t:$
\[
U_{t}+\lambda_{t}\mathbf{p}_{t}\cdot(\mathbf{q}-\mathbf{q}_{t})+\mu_{t}\mathbf{p}_{t}^{A}\cdot(\mathbf{q}^{A}-\mathbf{q}_{t}^{A})
\]
Because $\mathbf{p}_{t}\cdot\mathbf{q}_{t}\ge\mathbf{p}_{t}\cdot\mathbf{q}$,
\[
\lambda_{t}\mathbf{p}_{t}\cdot(\mathbf{q}-\mathbf{q}_{t})\le0
\]
Now consider two cases. Suppose first that $\mathbf{p}_{t}^{A}\cdot\mathbf{q}_{t}^{A}>x_{t}^{A}.$In
this case the labelling constraint does not bind and $\mu_{t}=0.$
Hence
\[
U_{t}\ge U_{t}+\underset{\le0}{\underbrace{\lambda_{t}\mathbf{p}_{t}\cdot(\mathbf{q}-\mathbf{q}_{t})}}+\underset{=0}{\underbrace{\mu_{t}}}\mathbf{p}_{t}^{A}\cdot(\mathbf{q}^{A}-\mathbf{q}_{t}^{A})
\]
 and thus $U(\mathbf{q}_{t})\ge U(\mathbf{q})$. 

\noindent Now consider the case in which the labelling constraint
binds and $\mu_{t}>0$. In that case $\mathbf{p}_{t}^{A}\cdot\mathbf{q}^{A}=x^{A_{t}}$
and 
\[
U_{t}\ge U_{t}+\underset{\le0}{\underbrace{\lambda_{t}\mathbf{p}_{t}\cdot(\mathbf{q}-\mathbf{q}_{t})}}+\underset{=0}{\mu_{t}\underbrace{\mathbf{p}_{t}^{A}\cdot(\mathbf{q}^{A}-\mathbf{q}_{t}^{A})}}
\]
 and thus $U(\mathbf{q}_{t})\ge U(\mathbf{q})$. $\blacksquare$

\section*{Proof of Proposition 4}

Implied by \cite{DiewertParkan1985}  and \cite{varian1983non}, Theorem 5. $\blacksquare$

\section*{Proof of Proposition 5}

Implied by \cite{varian1983non} Theorem 6.$%
\blacksquare$

\end{document}